\documentclass[12pt]{article}
\def\e{{\rm e}}
\def\a{\alpha}

\usepackage[english]{babel}
\usepackage{amssymb}
\usepackage{amsmath}
\usepackage{indentfirst}
\usepackage{graphicx}
\usepackage{graphics}
\usepackage{epsfig}
\usepackage{subfigure}
\usepackage{multirow}
\usepackage{color}
\usepackage{subfigure,a4}
\usepackage{epsfig,ulem}
\usepackage{dcolumn}

\newcommand{\be}{\begin{eqnarray}}
\newcommand{\ee}{\end{eqnarray}}

\usepackage[english]{babel}
\usepackage{amssymb}
\usepackage{amsmath}
\usepackage{amsfonts}
\usepackage{amsbsy}
\usepackage{indentfirst}
\usepackage{graphicx}
\usepackage{color}

\newcommand{\C}{\mathbb{C}}
\newcommand{\R}{\mathbb{R}}

\def\e{{\rm e}}

\newtheorem{theorem}{Theorem}[section]
\newtheorem{pro}{Proposition}[section]
\newtheorem{rmk}{Remark}[section]

\newenvironment{proof}[1][Proof]{\noindent\textbf{#1.} }{\ \rule{0.5em}{0.5em}}
\def\d{{\rm d}}

\def\a{\alpha}
\def\diag{{\rm diag}}
\def\<{\langle}
\def\>{\rangle}

\def\diag{{\rm diag}}

\newcommand{\beq}{\begin{equation}}
\newcommand{\eeq}{\end{equation}}

\newcommand{\bmat}{\begin{displaymath}}
\newcommand{\emat}{\end{displaymath}}

\def\1{{\bf 1}}

\def\ga{\gamma}
\begin{document}
\title{{Non self-adjoint
operators with real spectra and extensions of quantum mechanics}}
\author{
N. Bebiano\footnote{ CMUC, University of Coimbra, Department of
Mathematics, P 3001-454 Coimbra, Portugal (bebiano@mat.uc.pt)},
J.~da Provid\^encia\footnote{CFisUC, Department of Physics,
University of Coimbra, P 3004-516 Coimbra, Portugal
(providencia@teor.fis.uc.pt)}} \maketitle
\begin{abstract}
In this article, we review the general quantum mechanical setting
associated to a non self-adjoint Hamiltonian with  real spectrum.
Spectral properties of the Hamiltonian of a simple model
of the Swanson type are investigated.
The eigenfunctions associated to the real simple eigenvalues
are shown to form complete systems but not a (Riesz) basis,
which gives rise to difficulties in the rigorous mathematical formulation
of quantum mechanics. A new  inner product, which is appropriate for the physical interpretation of the model,
has been consistently introduced. The dynamics of the system is described.
Some specificities of the theory of non self-adjoint operators
with implications in quantum mechanics are discussed.
\end{abstract}
\section{Introductiom}
 In 
non-relativistic quantum mechanics,
the state of a particle is described, at the instant $t$, by a function
$\Psi_t(x)$, where $x$ denotes the particle coordinate. This function is
 called the {\it wave function}.
Its time evolution
is described  by the time dependent {\it Schr\"odinger equation},
$$i\frac{\partial \Psi_t(x) }{ \partial t}=H\Psi_t(x),$$
where $H$ is the {\it Hamiltonian operator} of the system. In most relevant cases $H$ acts on
an infinite dimensional separable Hilbert space
 $\cal H$,
endowed with the inner product $\langle\cdot,\cdot\rangle$ and corresponding norm $\|\cdot\|$.
The fundamental axiom of conventional formulations of quantum mechanics is that
$H$  is {\it Hermitian}, or synonimously, {\it selfadjoint}. That is, $H=H^*$, for $H^*$ the {\it adjoint} of $H$
$$\langle Hf,g\rangle:=\langle f,H^*g\rangle$$
for all $f,g\in {\cal H}$ such that $f\in{\cal D}(H)$ and $g\in{\cal D}(H^*)$.
Throughout, ${\cal D}(\cdot)$ will denote the domain of the operator under consideration.
Also the {\it observables} of the system are Hermitian, which ensures that the involved eigenvalues are real and
the corresponding eigenfunctions can be taken orthonormal in such a way that they form a basis of
the Hilbert space. As a consequence, meaningful
properties on the dynamics of the system follow.

A formal solution of the time dependent Schr\"odinger equation, which rules the system dynamics,
is  given by
$$\Psi_t(x)=\e^{-iHt}\Psi_0(x),$$
where $\Psi_0(x)$ is the wave function in the initial state $t=0.$
A simple consequence of the hermiticity of $H$ is the invariance of the norm of the wave function with time,
as $\exp(-iHt)$ is {\it unitary},
$$\|\Psi_t\|^2=
\langle\Psi_t,\Psi_t\rangle=
\langle\e^{-iHt}\Psi_0,\e^{-iHt}\Psi_0\rangle=
\langle\e^{iHt}\e^{-iHt}\Psi_0,\Psi_0\rangle=
\langle\Psi_0,\Psi_0\rangle.
$$
This property is physically important because it means that the number of particles
of the system does not change with time.
If the energy of a particle is measured in the state described by the wave function $\Psi(x)$,
the {\it expectation value} of the measurement, in a statistical sense, is given by the {\it Rayleigh quotient}
$$E:=\frac{\langle H\Psi,\Psi\rangle}{\langle\Psi,\Psi\rangle},$$
which is real if $H$ is Hermitian, as it should.


During the second half of last century, energy states of atoms,
molecules and atomic nuclei
have been usually described as eigenfunctions of selfadjoint {\it Schr\"odinger operators}.
The publication in 1998 by Bender and Boettcher of the seminal paper  on non-Hermitian Hamiltonians
with  $\mathrm{P}\mathrm{T}$-{\it symmetry} \cite{[1]},
where $\mathrm{P}$ and $\mathrm{T}$ are, respectively, the {\it parity} (or {\it space reflexion}) and the {\it time reversal} operators:
$$\mathrm{P}\Psi(x):=\Psi(-x),\quad \mathrm{T}\Psi(x):=\overline{\Psi(x)},$$
is a landmark. The development of $\mathrm{P}\mathrm{T}$-{\it symmetric quantum mechanics} was initiated and a growing literature
on $\mathrm{P}\mathrm{T}$-models found applications in different domains of physics.

Certain relativistic extensions of quantum
mechanics lead naturally to non-Hermitian Hamiltonian operators,
$H\neq H^*$. In this case, the Rayleigh quotient $\langle H\Psi,\Psi\rangle/\langle\Psi,\Psi\rangle$
does not provide the energy expectation value because in general it is not real, for $\Psi$ complex,
and the norm $\langle\exp(-iHt)\Psi_0,\exp(-iHt)\Psi_0\rangle$ becomes time dependent,
which is undesirable in the physical context. 
It became  fundamental to investigate
formulations of Quantum Mechanics for
non-Hermitian Hamiltonian operators, mathematically consistent and physically meaningful.
This objective has been the aim of intense
research activity in the last two decades. We refer to \cite{bagarello*} and references therein.
New  results opened new directions
 both in theoretical and experimental fronts, in classical and quantum domains.
Non-Hermitian Hamiltonians having real spectra exhibit a pathological behavior.
In the next section we illustrate the difficulties originated by these operators
in the development of mathematically rigorous quantum theories.

\subsection{Quasi-Hermitian QM}

The problem of how to construct a  consistent non-Hermitian quantum theory
has been investigated, mainly inspired by the knowledge that PT symmetric Hamiltonians possess
real spectra and allow for a unitary time evolution with a redefined inner product in the Hilbert space where the operator lives.
A necessary condition for developing such a theory is obviously the reality of the spectrum of the Hamiltonian, $\sigma(H)$,
but it is far from being sufficient.
In this context, there have been attempts to develop the so called {\it quasi-Hermitian quantum mechanics},
where the Hamiltonian $H$ is a {\it quasi-selfadjoint} operator,
that is, which satisfies the quasi-selfadjointness operator relation
\begin{equation}H^*\Theta=\Theta H,\label{H*Theta}\end{equation}
with $\Theta=T^*T$ a  positive, bounded and boundedly invertible operator, called a {\it metric}.
An operator $H$ with the above
property is actually Hermitian for the  new inner product
$$\ll \phi,\psi\gg:=\langle\Theta\phi,\psi\rangle=\langle T\phi,T\psi\rangle.
$$

The concept of quasi-selfadjointness, which
goes back to Dieudonn\'e \cite{dieudonne},
is of remarkable interest in the set up of non-Hermitian quantum mechanics.
A modified  inner product in the underlying Hilbert space,  relatively to which 
$H$ becomes selfadjoint via the similarity transformation $THT^{-1},$
\begin{equation}\label{h0}
\widetilde H=THT^{-1},
\end{equation}
where $\widetilde H$ is Hermitian, has been searched. If $T$ is bounded
and boundedly invertible, then the spectra of $THT^{-1}$ and $H$ coincide and the
eigenfunctions share basis properties.
Then, some fundamental issues of selfadjoint operators remain valid, such us spectral stability with respect to perturbations,
unitary evolution, etc.
It is not very common to find in
the literature non selfadjoint models for which such a metric is constructed,
neither
the existence of a metric operator is guaranteed.
Problems arise if $T$ or $T^{-1}$are unbounded, such as it may happen that the spectrum of $H$ is
discrete while $THT^{-1}$ has no eigenvalues.

In the finite dimensional setting, all the involved operators are bounded. In particular, if $T^{-1}$ exists it is automatically bounded, and
the concepts of quasi-Hermiticity and similarity to a selfadjoint operator work without difficulty,
because we are dealing essentially with finite matrices. The adjoint of $H^*$ is simply the transconjugate, the time evolution deduced from the Hamiltonian is unitary, and so it preserves the total probability of the system
given by $\int|\Psi_t(x)|^2\d x$ \cite{bagarello*}.

The rest of this note is organized as follows. In Section \ref{S2}, we consider a simple model
of the {\it Swanson type} \cite{swanson}, and review the general
quantum mechanical setting associated to a non selfadjoint Hamiltonian with a real spectrum.
In Section \ref{S3}, spectral properties of the Hamiltonian are investigated.
The eigenfunctions associated to the real simple eigenvalues
are shown to form complete systems but not a Riesz basis. In Section \ref{S4}, the dynamics of the system is described and
a new  inner product, which is
appropriate for the physical interpretation, is consistently introduced.
In Section \ref{S6}, some specificities of non selfadjoint operators
with implications in quantum mechanics are discussed
and useful mathematical background in this context is pointed out.

\section{The model}\label{S2}
We will be concerned with a model on the Hilbert space  ${\cal H}=L^2(\R)$ of square
integrable functions in one real variable, endowed with the standard
inner product
$$\langle
\Phi_\alpha,\Phi_\beta\rangle=\int_{-\infty}^{+\infty}
\Phi_\alpha(x)\overline{\Phi_\beta(x)} \d x,\quad \Phi_\alpha,\Phi_\beta\in L^2(\R).$$

The system we wish to study is a very simple model
of the {\it Swanson type} \cite{swanson}, characterized by the following Hamiltonian operator on $L^2(\R)$
\begin{equation}
H:=-\frac{1}{4}\frac{\partial^2}{\partial x^2}+x^2-\gamma\left(\frac{1}{2}+x\frac{\partial}{\partial x}\right),\quad \gamma\in\R\backslash\{0\},
\quad |\gamma|\leq1,
\end{equation}

Observe that, on $L^2(\R),$
\begin{equation}    \nonumber
H^*=-\frac{1}{4}\frac{\partial^2}{\partial x^2}+x^2+\gamma\left(\frac{1}{2}+x\frac{\partial}{\partial x}\right),
\end{equation}
so that $H\neq H^*$, for $\gamma\neq0.$ For $\gamma=0,$ it is obvious that $H=H^*.$

Most relevant operators in Quantum Mechanics are unbounded.
Unboundedness of operators in the infinite dimensional setting unavoidably restrict
their domains of definition to nontrivial subspaces of the Hilbert space.
The real parameter $\gamma$ must be carefully chosen so that the spectral theory of the operator
 can be developed in a rigourous mathematical framework.
The real parameter  $\gamma$, which {\it measures} the degree on non-Hermiticity of the Hamiltonian,
is assumed to be non-zero to avoid the well-known Hermitian case. We will consider $H$ as a perturbation
of the famous {\it harmonic oscillator} $H_{ho}$:
\begin{equation}    \nonumber
H_{ho}=-\frac{1}{4}\frac{\partial^2}{\partial x^2}+x^2,
\end{equation}
which coincides with $\R(H)=(H+H^*)/2.$
For our purposes, we impose the condition of smallness of $\gamma$, $|\gamma|<1,$
in order to ensure that the non-Hermitian term $V=-\gamma\left(\frac{1}{2}+x\frac{\partial}{\partial x}\right)$ does not completely change the behavior of $H_{ho}$.

The domain $\cal D$ of $H$ is
$${\cal D}:=\{\Psi(x)\in W^{1,2}(\R) :x^2\Psi(x)\in{L^2(\R)}\}.$$
Here, $W^{1,2}(\R)$ denotes the usual Sobolev space of functions on $L^2(\R)$ whose weak first and second derivatives belong to $L^2(\R).$
Observe that the domain contains the subspace $\cal S$ of functions $f(x)$ such that $\exp{(\gamma x^2)}f(x)\in{L^2(\R)},$
which in turn contains $C^\infty(\R),$
$${\cal S}:=\{\Psi(x)\in{\cal D} :\e^{\gamma x^2}\psi(x)\in{L^2(\R)}\}.$$
Thus, $H$ is densely defined in this domain, which ensures the existence and uniqueness of its adjoint $H^*$.
As $\Re(H)$ is closed, and $V$ is relatively bounded with respect to $\Re(H),$ with the relative bound smaller than 1, then $H$ is closed (\cite[Theorem 3.3]{gohberg}). The closedness of $H$ is a crucial starting point for the investigation of its spectrum,
because the spectrum is only meaningfully defined for closed operators.
We will show that $H$ has a purely discrete real spectrum if $|\gamma|$ is sufficiently small.

\subsection{A basis of $L^2(\R)$}
We also consider the auxiliary operator  $H_0:L^2(\R)\rightarrow L^2(\R),$
\begin{equation}
H_0:=-\frac{1}{4}\frac{\partial^2}{\partial x^2}+(1+\gamma^2)x^2.
\end{equation}

Notice that the following operator identity {\it formally} holds,
\begin{equation}\label{E3}
H_0=\e^{\gamma {x^2}}H\e^{-\gamma x^2},
\end{equation}
in the sense that the operators in the left and in the right hand sides act in the same manner on any
wave function $\Phi\in L^2(\R),$
\begin{equation}\nonumber
H_0\Phi(x)=\e^{\gamma {x^2}}H\e^{-\gamma x^2}\Phi(x).
\end{equation}
The word ``formal'' refers to the fact that $\exp(\gamma x^2)$ is unbounded.
We may also write the operator equality in (\ref{E3}) as,
\begin{equation}\nonumber
H=\e^{-\gamma {x^2}}H_0\e^{\gamma x^2},
\end{equation}
where it is implicitly assumed that the operators in  both sides
of the operator equality act on wave-functions $\Psi(x)\in{\cal S}\subset{\cal D}$. 
However, while $H_0$ goes from $L^2(\R)$ to $L^2(\R),$  $H$ goes from $\cal S$ to $\cal S$.

The spectrum and eigenvectors of $H_0$ are easily obtained
with the help of the {\it annihilation} bosonic operator
\begin{eqnarray*}
a:={(1+\gamma^2)^{1/4}}x+\frac{1}{2}~\frac{1}{(1+\gamma^2)^{1/4}}~\frac{\partial}{\partial x},
\end{eqnarray*}
and its adjoint, the {\it creation} bosonic operator,
\begin{eqnarray*}a^*:={(1+\gamma^2)^{1/4}}x-\frac{1}{2}~\frac{1}{(1+\gamma^2)^{1/4}}~\frac{\partial}{\partial x},
\end{eqnarray*}
which satisfy the commutation relation
$$[a,a^*]=aa^*-a^*a=\1,$$
where as usual, $\1$ denotes the identity operator.

The factorization of $H_0$ in terms of the bosonic operators is straightforwardly obtained,
\begin{equation}\nonumber
H_0=\sqrt{1+\gamma^2}~a^*a+\frac{1}{2}\sqrt{1+\gamma^2}~\1.
\end{equation}

The wave function
\begin{equation}
\nonumber
\Phi_0(x)=\e^{-{x^2/\sqrt{1+\gamma^2}}}\in L^2(\R),
\end{equation}
satisfies $a\Phi_0=0$ and describes the so called {\it groundstate} of $H_0$,
as it is
 an eigenfunction of $H_0$ associated with the lowest eigenvalue,
\begin{equation}
\nonumber
E_0=\frac{1}{2}\sqrt{1+\gamma^2}.
\end{equation}
The wave function
\begin{equation}
\Phi_n(x)=a^{*~n}\Phi_0(x),~n\geq0,
\nonumber
\end{equation}
is an eigenfunction of $H_0$ and describes the so called $n^{th}$ {\it bosonic state}. The associated eigenvalue is
\begin{equation}
E_n=\left(n+\frac{1}{2}\right)\sqrt{1+\gamma^2}, ~n=0,1,3,\ldots.
\nonumber
\end{equation}

The wave functions $\Phi_n(x)$ are orthogonal
\begin{equation}
\langle\Phi_n,\Phi_m\rangle=n!\delta_{nm}
\langle\Phi_0,\Phi_0\rangle,\quad m,n\geq0,
\nonumber
\end{equation}
for $\delta_{mn}$ the {\it Kronecker symbol} ($=1$ for $m=n$ and 0 otherwise),
and constitute a basis for $L^2(\R),$
\begin{equation}{\cal F}_\Phi:=\{\Phi_n={a^*}^n\Phi_0:n\geq0\},\label{bases}\end{equation}
as for any $\Phi\in{\cal H},$ there exists a set of complex coefficients such that $\Phi$
can be uniquely expressed as
$$\Phi=\sum_kc_k\Phi_k.$$

\subsection{Eigenvalues and Eigenfunctions of $H$}
Next, we consider the functions
\begin{equation}
\Psi_n(x)=\e^{-\gamma x^2}\Phi_n(x),
\nonumber
\end{equation}
which
belong to ${\cal S}\subset{\cal D}$, since we
clearly have  $\exp(\gamma x^2)\Psi_n(x)\in L^2(\R).$
If $$\gamma\geq-\sqrt{\frac{\sqrt5-1}{2}},$$
then, $\Psi_n(x)\in L^2(\R),$ because
$$\e^{-\gamma x^2}\Phi_0(x)=\e^{-\gamma x^2}\e^{-x^2/\sqrt{1+\gamma^2}}$$
 belongs to $L^2(\R)$ if
$(\gamma+1/\sqrt{1+\gamma^2})<1.$

Observing that
\begin{equation}
\nonumber
H\Psi_n(x)=H\e^{-\gamma x^2}\Phi_n(x)=\e^{-\gamma x^2}\e^{\gamma x^2}H\e^{-\gamma x^2}\Phi_n(x)=\e^{-\gamma x^2}H_0\Phi_n(x),
\end{equation}
we  obtain
\begin{equation}
H\Psi_n(x)=\left(n+\frac{1}{2}\right)~\sqrt{1+\gamma^2}\Psi_n(x).
\nonumber
\end{equation}
That is, $\Psi_n(x)$, for $n=0,1,2,\ldots,$ are eigenfunctions of $H$ and $(n+1/2)\sqrt{1+\gamma^2}$
are the associated eigenvalues.
If the  similarity relation (\ref{h0}) holds for a bounded and boundedly invertible $T$, the eigenvalues  of $H$
are the same of $H^*$. As in this case $T=\exp(\gamma x^2)$ is unbounded, this is not guaranted.
In the next section, this will be confirmed.
\subsection{Matrix representation of $H$}
Let us consider the bosonic operators
\begin{eqnarray*}
b:={x}{}+\frac{1}{2}~\frac{\partial}{\partial x},\quad
b^*:={x}{}-\frac{1}{2}~\frac{\partial}{\partial x},
\end{eqnarray*}
which satisfy the commutation relation
$$[b,b^*]=\1.$$
In terms of the bosonic operators, $H$ becomes
$$H=b^*b+\frac{\gamma}{2}({b^*}^2-b^2)+\frac{\1}{2}.$$
We notice that the {\it parity operator} $\mathrm{P}$ commutes with $H$, and
the eigenspaces of $\mathrm{P}$ are invariant subspaces of $H$.

With respect to the basis constituted by the eigenfunctions of the {\it number operator} $N:=b^*b$,
$${\cal F}_\phi=\{\phi_n={b^*}^n\phi_0:~b\phi_0=0,~n\geq0\},$$
the operator $b$ is represented by the upper shifted matrix
\begin{eqnarray*}
B=\left[\begin{matrix}
0&\sqrt1&0&0&\ldots\\
0&0&\sqrt2&0&\ldots\\
0&0&0&\sqrt3&\ldots\\
\vdots&\vdots&\vdots&\vdots&\ddots\end{matrix}\right],
\end{eqnarray*}
and the operator $b^*$ is represented by the tanspose of $B$, $B^T.$ These matrices satisfy the commutation relation,
$$[B,B^T]=I.$$

The matrix $B^T$ is a {\it raising matrix} because, if $\Phi$ is an eigenvector of $A_0$
associated with the eigenvalue $\Lambda$,
$$A_0\Phi=\Lambda\Phi,$$
then $B^T\Phi$ is an eigenvector of $A_0$ associated with the upwardly shifted eigenvalue
$\Lambda+1,$
$$A_0B^T\Phi=(\Lambda+1)B^T\Phi,$$
and similarly $B$ is a {\it lowering matrix}, as
$$A_0B\Phi=(\Lambda-1)B\Phi,$$
if $ B\Phi\neq0.$

In the same basis, the operator ${b^*}^2$ is represented by the matrix $A_+$,
$$A_+=\left[\begin{matrix}
0&0&0&\ldots\\
0&0&0&\ldots\\
\sqrt{1\times2}&0&0&\ldots\\
0&\sqrt{2\times3}&0&\ldots\\
0&0&\sqrt{3\times4}&\ldots\\
\vdots&\vdots&\vdots&\ddots
\end{matrix}\right]$$
the operator ${b}^2$ is represented by the matrix $A_-={A_+}^T$ and
the operator ${b^*}b$ is represented by the matrix $A_0={\rm diag}(0,1,2,3,\ldots).$
Notice that
$$[A_0,A_+]=2A_+,\quad [A_0,A_-]=-2A_-.$$

Thus, $H$ is represented by the pentadiagonal matrix
\begin{eqnarray*}&&B^TB+\frac{\gamma}{2}((B^T)^2-B^2)+\frac{I}{2}=A_0+\frac{\gamma}{2}(A_+-A_-)+\frac{I}{2}\\
&&=\left[\begin{matrix}
1/2&0&-\frac{\gamma}{2}\sqrt{1\times2}&0&\ldots\\
0&3/2&0&-\frac{\gamma}{2}\sqrt{2\times3}&\ldots\\
\frac{\gamma}{2}\sqrt{1\times2}&0&5/2&0&\ldots\\
0&\frac{\gamma}{2}\sqrt{2\times3}&0&7/2&\ldots\\
0&0&\frac{\gamma}{2}\sqrt{3\times4}&0&\ldots\\
0&0&0&\frac{\gamma}{2}\sqrt{4\times5}&\ldots\\
\vdots&\vdots&\vdots&\vdots&\ddots
\end{matrix}\right],\end{eqnarray*}
which may be written as
\begin{eqnarray*}\left[\begin{matrix}
1/2&-\frac{\gamma}{2}\sqrt{1\times2}&0&\ldots\\
\frac{\gamma}{2}\sqrt{1\times2}&5/2&-\frac{\gamma}{2}\sqrt{3\times4}&\ldots\\
0&\frac{\gamma}{2}\sqrt{3\times4}&0&\ldots\\
\vdots&\vdots&\vdots&\ddots
\end{matrix}\right]\oplus\left[\begin{matrix}
3/2&-\frac{\gamma}{2}\sqrt{2\times3}&0&\ldots\\
\frac{\gamma}{2}\sqrt{2\times3}&7/2&-\frac{\gamma}{2}\sqrt{4\times5}&\ldots\\
0&\frac{\gamma}{2}\sqrt{4\times5}&0&\ldots\\
\vdots&\vdots&\vdots&\ddots
\end{matrix}\right].
\end{eqnarray*}
The eigenfunctions of $N$, $\phi_0(x),$ $\phi_2(x),\phi_4(x),
\cdots$ are even functions, while $\phi_1(x),$ $\phi_3(x), \phi_5(x), \cdots$ are odd functions.
The eigenspaces of $\mathrm{P}$
are invariant subspaces of $H$.
Thus, $H$ is represented by
real tridiagonal matrices called {\it pseudo-Jacobi matrices} in the bases
of these subspaces, that is, Jacobi matrices pre multiplied by
$J={\rm diag}(1,-1,1,-1,\ldots).$

Next, in order to determine raising and lowering matrices for $(A_0+(A_+-A_-)/2+{I}/{2}), $
we look for linear combinations of $B^T$ and $B$ satisfying
$$\left[\left(B^TB+\frac{\gamma}{2}(({B^T})^2-B^2)+\frac{I}{2}\right),(xB^T+yB)\right]=\lambda(xB^T+yB).$$
By some computations, we find that $\lambda=\pm\sqrt{1+\ga^2}$
and we obtain the sought for lowering and raising  matrices
\begin{eqnarray*}
&&D=
\frac{1}{2}\left(({1+\ga^2})^{1/4}+\frac{1+\ga}{({1+\ga^2})^{1/4}}\right)B+
\frac{1}{2}\left(({1+\ga^2})^{1/4}-\frac{1-\ga}{({1+\ga^2})^{1/4}}\right)B^T,\\
&&D^\ddag=
\frac{1}{2}\left({{1+\ga^2}}^{1/4}-\frac{1+\ga}{({1+\ga^2})^{1/4}}\right)B+
\frac{1}{2}\left(({1+\ga^2})^{1/4}+\frac{1-\ga}{({1+\ga^2})^{1/4}}\right)B^T.
\end{eqnarray*}
These matrices satisfy the following commutation relations,
$$[D,D^\ddag]=I.$$
$$\left[\left(A_0+\frac{\gamma}{2}(A_+-A_-)+\frac{I}{2}\right),D^\ddag\right]=\sqrt{1+\ga2}D^\ddag,$$
$$
\left[\left(A_0+\frac{\gamma}{2}(A_+-A_-)+\frac{I}{2}\right),D\right]=-\sqrt{1+\ga^2}D.$$
This means that, if $\Upsilon$ is an eigenvector of $(A_0+{\gamma}/{2}(A_+-A_-)+{I}/{2})$
associated with the eigenvalue $\Lambda$, then
$D^\ddag\Upsilon$ and $D\Upsilon$ are eigenvectors associated, respectively,
with the eigenvalues $\Lambda+\sqrt{1+\ga^2}$ and  $\Lambda-\sqrt{1+\ga^2}$.
Moreover,
\begin{equation}\nonumber
A_0+\frac{\gamma}{2}(A_+-A_-)+\frac{I}{2}=\sqrt{1+\gamma^2}~D^\ddag D+\frac{1}{2}\sqrt{1+\gamma^2}~I.
\end{equation}
Thus
$$\sigma\left(A_0+\frac{\gamma}{2}(A_+-A_-)+\frac{I}{2}\right)=\left(n+\frac{1}{2}\right)\sqrt{1+\ga^2},~~n\geq0.$$

An eigenvector $ \Upsilon_0$ of $\left(A_0+\frac{\gamma}{2}(A_+-A_-)+{I}/{2}\right)$ associated with the lowest eigenvalue
$\sqrt{1+\ga^2}$ is such that
$$D\Upsilon_0=0.$$
We find
$$\Upsilon_0=\left[1,0,\sqrt{1\over2}~\eta,0,\sqrt{1\times3\over2\times4}~\eta^2,0,\sqrt{1\times3\times5\over2\times4\times6}~\eta^3,0,\ldots\right]^T,$$
where
$$\eta=\frac{1-\ga-\sqrt{1+\ga^2}}{1+\ga+\sqrt{1+\ga^2}}.$$
An eigenvector of $\left(A_0+{\gamma}/{2}(A_+-A_-)+{I}/{2}\right)$ associated with the  eigenvalue
$(n+1/2)\sqrt{1+\ga^2}$ is given by
$$\Upsilon_n={D^\ddag}^n\Upsilon_0.$$
\color{black}


\section{Spectrum of H}\label{S3}
The spectrum of an operator on a finite dimensional
Hilbert space is exhausted by the eigenvalues, but, in the infinite dimensional setting,
there are additional parts of the spectrum of $H$ to be considered.

The {\it resolvent set} of $H$, denoted by $\rho(H)$, is constituted
by all the complex numbers $\lambda$ for which
$
(H-\lambda)^{-1}$
exists as a bounded operator on $\cal H$. The {\it spectrum} of $H$ is the complement of the resolvent set
$$\sigma(H)=\C\backslash\rho(H).$$
The set of all eigenvalues of $H$ is the
{\it point spectrum}, 
denoted by $\sigma_p(H),$ and is
formed by the complex numbers $\lambda$ for which
$H-\lambda:{\cal D}(H)\rightarrow {\cal H}$ is not injective.
The {\it continuous spectrum} is constituted by those $\lambda$ such that $H-\lambda$ is injective
and its range is dense.
The  {\it residual
spectrum} consists of  those $\lambda$ for which $H-\lambda$ is injective and its range
is not dense. The spectrum  $\sigma(H)$ is the union of
these three disjoint spectra.

The  spectrum of selfadjoint operators is non empty, real,
and the residual spectrum is empty, while the spectrum of non
selfadjoint operators can be empty or coincide with the whole
complex plane (see, e.g., refs. \cite{siegle,[3]}).
As already mentioned, the spectrum of an operator is meaningfully defined only for closed operators,
that is, those for which the set
$$\{\langle H\psi,\psi\rangle:~\psi\in{\cal D}\}$$
is a linear closed subspace of ${\cal H}\times{\cal H}.$ In Subsection \ref{ss3.2}, we show that $H$
is closed and that $\sigma(H)$
reduces to the point spectrum.
For this purpose we next introduce the central auxiliary concept of numerical range.

\subsection{Numerical range}
The {\it numerical range} of $H$ is denoted and defined as
$$W(H):=\{\langle H\psi,\psi\rangle:\psi\in{\cal D}(H),\|\psi\|=1\}.$$
In general $W(H)$ is neither open nor closed, even when $H$ is a closed
 operator. For $H$ bounded the following spectral inclusion holds:
 $$\sigma(H)\subset\overline{W(H)}.$$
\begin{theorem}\label{P3.1}
The numerical range of $H$ is bounded by the hyperbola branch
\begin{eqnarray*}
 &&\left(x-\frac{1}{2}\right)^2-\frac{y^2}{4\ga^2}=\frac{1}{4},\quad x\geq1.
\end{eqnarray*}
\end{theorem}
\begin{proof}
By the Toeplitz-Hausdorff Theorem, the numerical range of $H$ is  convex. So, let us consider the supporting line of $W(H)$
perpendicular to the direction $\theta$. The distance
of this line to the origin is the lowest eigenvalue of
$$\Re (\e^{-i\theta}H)=b^*b\cos\theta-i\frac{\gamma}{2}({b^*}^2-b^2)\sin\theta+\frac{\cos\theta}{2},$$
provided this operator is bounded from below, which occurs for $-\pi/2\leq\theta<\pi/2$.
The eigenvalues of  $\Re(\e^{-i\theta}H)$  are readily determined by
the EMM \cite{davydov}, and they are found to be
$$
(n+ \frac{1}{2}) \left(\cos\theta - \sqrt{\cos^2\theta - 4 \ga^2 \sin^2\theta}\right),~n\geq0.$$
Thus, the considered supporting line has equation
\begin{equation}x\cos\theta+y\sin\theta=
 \frac{1}{2} \left(\cos\theta - \sqrt{\cos^2\theta - 4 \ga^2 \sin^2\theta}\right).\label{support}\end{equation}
The boundary equation of  $W(H)$ is found eliminating $\theta$
between (\ref{support}) and
\begin{equation}\nonumber -x\sin\theta+y\cos\theta=
 \frac{1}{2}\frac{\d}{\d x} \left(\cos\theta - \sqrt{\cos^2\theta - 4 \ga^2 \sin^2\theta}\right).\end{equation}
\end{proof}\color{black}
\begin{figure}[h]
\centering
\includegraphics[width=.5\textwidth, height=0.5\textwidth]
{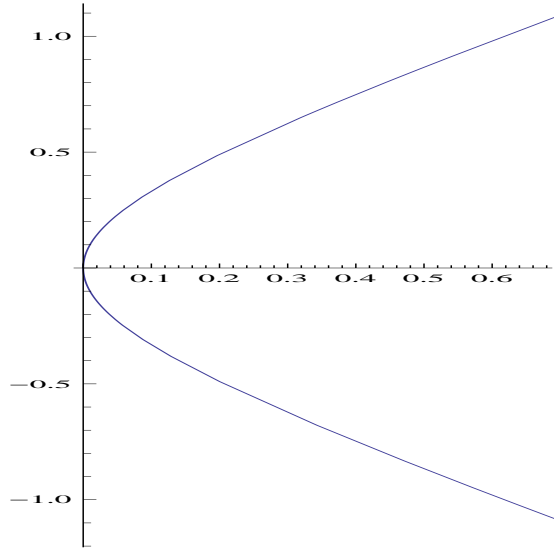} \caption{
Numerical ranges of $H$, $\gamma=1/2$.}
\label{fig011}
\end{figure}

\subsection{Accretivity of $H$}\label{ss3.2}

 An operator is said to be {\it accretive} if 
its numerical range is a subset of  the sector with vertex at the origin
  and {\it semi-angle} $0\leq\omega<\pi/2$,
$$S_{0,\omega}=\{z\in\C: 0\leq|\arg z|\leq\omega\}.$$
An operator $H$ is $m$-{\it accretive} if
its numerical range is contained in the closed right half-plane
and the so called {\it resolvent bound} 
holds:
$$\forall\lambda\in\C,~\Re\lambda<0,~\|(H-\lambda)^{-1}\|\leq1/|\Re \lambda|.$$
\begin{theorem}
The operator $H$ is $m$-accretive.
\end{theorem}
\begin{proof} Obviously, the operator $H$ is 
{\it accretive} because
$W(H)\subset S_{0,\pi/2}$.
We show that for any $z\in\C,$ with $\Re z<0$, the resolvent bound holds.
We have
$${\rm dist}(z,\overline{W(H)})\leq|\langle H\psi,\psi\rangle-z|=
|\langle(H-z)\psi,\psi\rangle|\leq\|(H-z)\psi\|.$$
As ${\rm dist}(z,\overline{W(H)})\geq|\Re z|,$
having in mind Theorem \ref{P3.1}, the result follows.
\end{proof}

The closed operator $H$ on $\cal H$  has a {\it compact resolvent} if
$\rho(H)\neq{\O}$ and the inverse operator, $(H-\lambda)^{-1}$, for some $\lambda\in\rho(H),$
is  compact. Notice that
$\Re(H)$ is an $m$-accretive
operator, since $\Re(H)$ is Hermitian and $W(\Re(H))$ lies on the positive real axis.

Moreover,
$\Re(H)$ is a closed operator, has a compact resolvent and since the perturbation operator $V$ is relatively bounded with respect to $\Re(H)$
with relative bound smaller than $1$, then $H=\Re(H)+\lambda V$ has a
compact resolvent \cite[Theorem 5.4.1]{Krejcirik*}. Now, by %
\cite[Theorem IX, 2.3]{Krejcirik*},
if $H$ has a compact resolvent, then $\sigma(H)=\sigma_p(H)$.
\subsection{Pseudospectrum}
The set of the eigenfunctions of a selfadjoint operator with purely discrete real spectrum
can be orthonormalized so that it forms an orthonormal basis. Eigenfunctions of a non-Hermitian operator $H$ are typically not orthogonal.
The eigenfunctions  of $H$, that has a purely discrete real spectrum, form a {\it Riesz basis} if $H$ is quasi-Hermitian
(see (\ref{H*Theta}))
with bounded and boundedly invertible metric $\Theta$. Riesz basicity is not preserved by an unbounded operator.

Our objective
is to show that the eigenfunctions $\Psi_n$ do not form a Riesz basis,
and so any metric $\Theta$ in (\ref{H*Theta}) is necessarily {\it singular},
that is, no bounded metric with bounded inverse exists.

For this purpose, we consider the {\it$\epsilon$-pseudospectrum} of $H$, $\epsilon>0,$ denoted and defined as follows
$$\sigma_\epsilon(H):=\{z\in\C:\|(H-z)^{-1}\|>\epsilon^{-1}\}$$
with the convention $\|(H-z)^{-1}\|=\infty$ for $z\in\sigma(H).$ The $\epsilon$-pseudospectrum always contains an
$\epsilon$-neighborhood of the spectrum. If the operator is selfadjoint, equality holds and $H$ is said to have a {\it trivial} pseudospectrum.
A non selfadjoint operator has a typically much larger pseudospectrum, which means that very small perturbations may drastically change the spectrum.

Numerical computations
carried out with {\it Matlab} show  that the pseuspectrum of $H$ is far from being trivial.
A non trivial pseudospectrum ensures the non existence of a bounded metric \cite{siegle}.
\color{black}
\subsection{Completeness of eigenfunctions}
We will show that the eigenfunctions of $H$ form a complete set in $L^2(R).$
Completeness of the system $\{\Psi_n\}$ means that its span is dense in $L^2(\R).$
A basis is complete, but the converse may not be true.

The $m$-accretivity of $H$ implies that $-iH$ is dissipative, i.e.,
$$\Im\langle H\Psi,\Psi\rangle\leq 0,~\forall \Psi\in{\cal D}(H).$$
As a consequence, the imaginary part of $(-iH-\epsilon)^{-1}$, for $\epsilon<0$, is
non-negative,
$$\frac{1}{2i}((-iE-\epsilon)^{-1}-(iH^*-\epsilon)^{-1})\geq0.$$
Applying \cite[Theorem VII,8.1]{gohberg}, the completeness of $\{\Psi_n\}$ follows.
\section{The dynamics of the system}\label{S4}
\subsection{ $\cal G$-quasi bases}
\bigskip

Let $\widetilde\Psi_n(x)$ be the eigenfunction of $H^*$ sharing with
$\Psi_n(x)$ the same eigenvalue of $H$.
The set
$${\cal F}_{\widetilde\Psi}=\{\widetilde\Psi_k(x):k\geq0\}$$
is complete but not necessarily a basis.
It is known \cite{siegle} that eigenfunctions of an operator  $H$ with
purely discrete spectrum form a Riesz basis if and only if $H$ is quasi-Hermitian
with bounded and boundedly invertible metric.

The eigenfunctions $\Psi_n(x),~\widetilde\Psi_m(x),~ m,n=0,1,2,3,\ldots,$
constitute biorthogonal systems, as
$$\langle\Psi_m,\widetilde\Psi_n\rangle=\delta_{mn}.$$
For the physical interpretation of the model, we consider the following subset of $\cal H$:
$${\cal D}_{phys}=\left\{\psi(x)\in{\cal H}:\psi(x)=\sum_{k=0}^\infty c_k~\Psi_k(x),~~c_k\in\C\right\}.$$
This set contains all the physically relevant wave functions of the physical system.
Even if ${\cal F}_\Psi$ is not a basis of $\cal H$, it is sufficient to expand physically meaningful wave functions.

Let
$${\cal G}:={\rm span}\{\Psi_n(x)\}\cap {\rm span}\{\widetilde\Psi_n(x)\}.$$
The following {\it resolution of the identity} holds, for any $f,g\in{\cal G}.$
\begin{eqnarray*}
&&\langle f,g\rangle
=\sum_{m,n}{\langle f,\widetilde\Psi_{n}\rangle\langle\Psi_{n},g\rangle\over
\langle\Psi_{n},\widetilde\Psi_{n}\rangle}
=\sum_{m,n}{\langle f,\Psi_{n}\rangle\langle\widetilde\Psi_{n},g\rangle\over
\langle\widetilde\Psi_{n},\Psi_{n}\rangle}.
\end{eqnarray*}
Following Bagarello, we say that the wave functions
$\Psi_{n}$ and $\widetilde\Psi_{n}$ are $\cal G$-{\it quasi basis} \cite{bagarello1,bagarello2}.

For our purposes it is enough to consider
\begin{eqnarray*}
&&\langle f,g\rangle
=\sum_{m,n}{\langle f,\widetilde\Psi_{n}\rangle\langle\Psi_{n},g\rangle\over
\langle\Psi_{n},\widetilde\Psi_{n}\rangle},
\end{eqnarray*}
for $f\in{\rm span}\{\Psi_n(x)\},~g\in{\rm span}\{\widetilde\Psi_n(x)\},$ a situation which Bagarello also
envisages.

The inner product $\langle\cdot,\cdot\rangle$ is not adequate for expressing the
conservation of the {\it particle number} and for the computation of the expectation value of energy measurements,
as it may yield complex values.
For functions $\Psi_\alpha(x),\Psi_\beta(x)\in\cal S,$ 
we consider
the {\it physical  inner product} defined by
\begin{equation}
\ll\Psi_\alpha,\Psi_\beta\gg=\langle\e^{\gamma x^2}\Psi_\alpha,\e^{\gamma x^2}\Psi_\beta\rangle
=\int_{-\infty}^{+\infty}\e^{2\gamma x^2}\Psi_\alpha(x)\overline{\Psi_\beta(x)}\d x.
\nonumber
\end{equation}
Following Mostazadeh \cite{mostafa}, we say that the {\it physical Hilbert space}
is the space of the functions $\Psi(x)\in {\cal D}_{phys}$ 
endowed with the inner product $\ll\cdot,\cdot\gg$.

The Hamiltonian $H$ is symmetric with respect to the inner product $\ll \cdot,\cdot\gg$, for wave functions
$\Psi(x)\in {\cal D}_{phys}$, because
\begin{eqnarray*}
&&\ll H\Psi,\Psi\gg=\int_{-\infty}^{+\infty}\e^{2\gamma^2}(H\Psi(x))\overline{\Psi(x)}\d x\\
&&
=\int_{-\infty}^{+\infty}(\e^{\gamma^2}H\Psi(x))\overline{\e^{\gamma x^2}\Psi(x)}\d x\\
&&
=\int_{-\infty}^{+\infty}(H_0\e^{\gamma^2}\Psi(x))\overline{\e^{\gamma x^2}\Psi(x)}\d x\\
&&
=\int_{-\infty}^{+\infty}\e^{\gamma^2}\Psi(x)\overline{H_0\e^{\gamma x^2}\Psi(x)}\d x,
\\&&
=\int_{-\infty}^{+\infty}\e^{\gamma^2}\Psi(x)\overline{\e^{\gamma x^2}H\Psi(x)}\d x
=\ll\Psi,H\Psi\gg.
\end{eqnarray*}
Since
$\ll H\Psi,\Psi\gg=\ll\Psi,H\Psi\gg$
for wave functions $\Psi(x)\in{\cal D}_{phys}$, 
it is clear that
the Rayleigh quotients of $H$ are real for the inner product $\ll\cdot,\cdot\gg$,
$$\frac{\ll H\Psi,\Psi\gg}{\ll\Psi,\Psi\gg}\in\R$$
and that a time-invariant norm is obtained
$$\ll\e^{-iHt}\Psi,\e^{-iHt}\Psi\gg=\ll\Psi,\Psi\gg.$$

Let us consider now the pseudo-bosonic operators
\begin{eqnarray*}&&
d:=\e^{-\gamma x^2}a\e^{\gamma x^2}={(1+\gamma^2)^{1/4}}x+\frac{1}{2}~\frac{1}{(1+\gamma^2)^{1/4}}
\left(\frac{\partial}{\partial x}+2\gamma x
\right)
,\\&&
d^\ddag:=\e^{-\gamma x^2}a^*\e^{\gamma x^2}={(1+\gamma^2)^{1/4}}x-\frac{1}{2}~ \frac{1}{(1+\gamma^2)^{1/4}}
\left(\frac{\partial}{\partial x}+2\gamma x
\right).
\end{eqnarray*}
These operators are called {\it pseudo-bosonic} because $b^\ddag\neq b^*$, with respect to the inner product $\langle\cdot,\cdot\rangle$.
The factorization of $H$ in terms of the pseudo-bosonic operators is straightforward:
\begin{equation}\nonumber
H=\sqrt{1+\gamma^2}~d^\ddag d+\frac{1}{2}\sqrt{1+\gamma^2}~\1.
\end{equation}
With respect to the inner product $\ll\cdot,\cdot\gg$,  $d^\ddag$ is the adjoint of $d$.


The physical inner product is appropriate to characterize the transition
probability amplitude from the state $\Psi_\alpha$ to the state $\Psi_\beta,$
\begin{equation}A_{\Psi_\alpha\rightarrow\Psi_\beta}=\frac{\ll\Psi_\alpha,\Psi_\beta\gg}{\sqrt{\ll\Psi_\alpha,\Psi_\alpha\gg
\ll\Psi_\beta,\Psi_\beta\gg}}~.\label{eq*}\end{equation}

\subsection{ Compression of $H$ to a subspace}

In general, the procedure followed in the analysis of the present model
may lead to 
problems which may be avoided by considering a linear combination of a finite number of eigenfunctions $\Psi_i(x)$ of $H$.
This idea is physically attractive in view of the superposition principle, and  has been
put by Bagarello \cite{bagarello2} in a mathematically consistent context.

We assume that $\Psi_i(x)\in L^2(\R)$.
Let $\widehat Q^{-1}\in\C^{n\times n}$ be the matrix  defined by
\begin{eqnarray*}
(\widehat Q^{-1})_{j,i}=\langle\Psi_i,\Psi_j\rangle,\quad i,j=0,1,2,\ldots,n-1.
\end{eqnarray*}
The matrix $\widehat Q^{-1}$ is  positive definite, because
\begin{eqnarray*}
\langle\Psi _0z_0+\ldots+\Psi_{n-1}z_{n-1},\Psi_0z_0+\ldots+ \Psi_{n-1}z_{n-1}\rangle>0,\quad\text{\rm for any}\quad(z_0,\ldots, z_{n-1})\neq0.
\end{eqnarray*}
Let ${\cal H}_n$ be the subspace of $\cal H$ spanned by the wave functions  $\Psi_i(x),~i=0,\ldots,n-1.$
An orthonormal basis of this finite dimensional space  is constituted by the wave functions
$$\Phi_j(x)=\sum_{k=0}^{n-1}(\widehat Q^{1/2})_{kj}\Psi_k(x),~j=0,\ldots,n-1.$$
In ${\cal H}_n,$ the Hamiltonian $H$ is represented by the matrix $\widehat H\in\C^{n\times n}$ with entries
\begin{eqnarray*}
&&\widehat H_{ji}=\langle H\Phi_i,\Phi_j\rangle\\
&&=\sum_{k=0}^{n-1}\sum_{l=0}^{n-1}(\widehat Q^{1/2})_{ki}\overline{(\widehat Q^{1/2})_{lj}}\langle H\Psi_k,\Psi_l\rangle\\
&&=\sum_{k=0}^{n-1}\sum_{l=0}^{n-1}(\widehat Q^{1/2})_{ki}{(\widehat Q^{1/2})_{jl}}\lambda_k\langle\Psi_k,\Psi_l\rangle\\
&&=\sum_{k=0}^{n-1}\sum_{l=0}^{n-1}(\widehat Q^{1/2})_{ki}{(\widehat Q^{1/2})_{jl}}\lambda_k(\widehat Q^{-1})_{lk}\\
&&=\sum_{k=0}^{n-1}{(\widehat Q^{-1/2})_{jk}}\lambda_k(\widehat Q^{1/2})_{ki},
\end{eqnarray*}
where $\lambda_0,\ldots,\lambda_{n-1}$ are the eigenvalues of $H$ associated, respectively, with the eigenfunctions
$\Psi_0(x),\ldots,\Psi_{n-1}(x)$. The matrix $\widehat H$ is non-Hermitian because
$$
\widehat H_{ji}=\sum_{k=0}^{n-1}{(\widehat Q^{-1/2})_{jk}}\lambda_k(\widehat Q^{1/2})_{ki}
\neq\sum_{k=0}^{n-1}{(\widehat Q^{1/2})_{ik}}\lambda_k(\widehat Q^{-1/2})_{kj}=
\overline{\widehat H_{ij}}=(\widehat H^*)_{ji}.$$
 The matrix $\widehat H$ is {\it pseudo-Hermitian}, since
$$\sum_{h=0}^{n-1}\widehat Q_{jh}\widehat H_{hi}=\sum_{h=0}^{n-1}(\widehat H^*)_{jh}\widehat Q_{hi},
$$
that is
$$\widehat Q\widehat H=\widehat H^*\widehat Q.$$

Let us consider the vectors
\begin{eqnarray*}
\widehat E_0=(1,0,\ldots,0)^T,\ldots,\widehat E_{n-1}=(0,\ldots,0,1)^T\in\C^n,
\end{eqnarray*}
and also the vectors
\begin{eqnarray*}
\widehat\Psi_i=\widehat Q^{-1/2}\widehat E_i,\quad i=0,\ldots, n-1.
\end{eqnarray*}
Since
$$\widehat H=\widehat Q^{-1/2}{\rm diag}(\lambda_0,\ldots,\lambda_{n-1})\widehat Q^{1/2}$$
it follows that
$$\widehat H\widehat \Psi_i=\lambda_i\widehat \Psi_i,\quad i=0,\ldots,n-1.$$
Notice that
$$\langle\widehat \Psi_i,\widehat \Psi_j\rangle=(\widehat Q^{-1})_{ji}.$$
The vectors
\begin{eqnarray*}
\widehat{\widetilde\Psi}_i=Q^{1/2}\widehat E_i,\quad i=0,\ldots, n-1,
\end{eqnarray*}
are orthonormal to the eigenvectors $\widehat \Psi_j,$
\begin{eqnarray*}
\langle\widehat{\widetilde\Psi}_i,\widehat\Psi_j\rangle=\delta_{ij}.
\end{eqnarray*}
These vectors are eigenvectors of $\widehat H^*$, and the functions
$$\widetilde\Psi_i(x)=d^{*~i}\widetilde\Psi_0(x),\quad d^{\ddag~*}\widetilde\Psi_0(x)=0,$$
are eigenfunctions of $H^*.$ In fact, we have
$ H^*{\widetilde\Psi}_i=\lambda_i{\widetilde\Psi}_i$ and
$$\widehat H^*\widehat{\widetilde\Psi}_i=\lambda_i\widehat{\widetilde\Psi}_i.$$

The standard inner product $\langle\cdot,\cdot\rangle$ in $\C^n$, is given by
$$\langle\widehat\Psi_\alpha,\widehat\Psi_\beta\rangle=
\sum_{j=0}^{n-1}\overline\beta_j\alpha_j
$$
where
$\widehat\Psi_\alpha=(\alpha_1,\ldots,\alpha_n)^T\in\C^n,~~\widehat\Psi_\beta=(\beta_1,\ldots,\beta_n)^T\in\C^n.$
We also introduce the physical inner product $\ll\cdot,\cdot\gg$
defined by
$$\ll\widehat\Psi_\alpha,\widehat\Psi_\beta\gg:=\langle\widehat Q\widehat\Psi_\alpha,\widehat\Psi_\beta\rangle=\sum_{i,j=0}^{n-1}
\widehat Q_{ij}\alpha_j\overline\beta_i.$$
We may now define the {\it energy expectation value} as
$$E=\frac{\ll H\widehat\Psi,\widehat\Psi\gg}{\ll\widehat\Psi,\widehat\Psi\gg}\in\R$$
and the {\it transition probability amplitude} as
\begin{equation}A_{\widehat\Psi_\alpha\rightarrow\widehat\Psi_\beta}=\frac{\ll\widehat\Psi_\alpha,\widehat\Psi_\beta\gg}
{\sqrt{\ll\widehat\Psi_\alpha,\widehat\Psi_\alpha\gg
\ll\widehat\Psi_\beta,\widehat\Psi_\beta\gg}}~.\label{eq**}\end{equation}
The expressions (\ref{eq*}) and (\ref{eq**}) are in consonance, but (\ref{eq*}) applies to functions
and (\ref{eq**}) to vectors in the compressed subspace.

\begin{rmk}
The non Hermitian matrix $\widehat H$ is similar to the Hermitian matrix
${\rm diag}(\lambda_0,\ldots\lambda_n)$,
and the space spanned by the eigenvectors $\widehat\Psi_0,\ldots,\widehat\Psi_n$ of $\widehat H$
coincides with the space spanned by the eigenvectors $\widehat E_0,\ldots,\widehat E_n$ of
${\rm diag}(\lambda_0,\ldots\lambda_n)$.
On the other hand, because the operator $\exp(\gamma x^2)$ is unbounded,
the non Hermitian operator $ H$ is not similar to the Hermitian operator
$H_0$,
and the space spanned by the eigenfunctions $\Psi_0,\Psi_1,\ldots$ of $ H$ does not
coincide with the space spanned by the eigenfunctions $\Phi_0,\Phi_n,\ldots$ of
$H_0$.
\end{rmk}

\section{Final remarks}\label{S6}
Viewing the operator $H$ as the Hamiltonian of a physical model, problems arise
from non Hermiticity.
The original inner product, defined in the Hilbert space $\cal H$ where $H$ lives, is not adequate for
the physical interpretation of the model.
A new inner product, which is
appropriate for that purpose, has been introduced.
{Although $H$ is non-Hermitian with respect to the initial inner product $\langle\cdot,\cdot\rangle$,
it becomes symmetric with respect to the physical inner product $\ll\cdot,\cdot\gg=\langle \Theta\cdot,\cdot\rangle$.}

We  define a subspace of
the Hilbert space such that the compression of the Hamiltonian operator
to that  subspace shares part of the spectrum and eigenfunctions
of the original one.
The referred subspace remains invariant under
the action of $H$.
However, stating that a Hermitian operator represents in a
reasonable sense the non-Hermitian operator may be controversial,
since relevant information on the Hamiltonian may not be captured in the subspace where it lives.

The concept of pseudospectrum is of great relevance for the description
of non-Hermitian operators in the context of QM, as well as the one of numerical range.
Non-Hermitian operators have typically non-trivial pseudospectra.
If the quasi-Hermiticity relation
holds with a positive bounded and boundedly invertible
metric, then the pseudospectrum of $H$ is trivial. A non-trivial pseudospectrum determines the
non-existence of such a metric.
Moreover, it implies the non existence of a Riesz basis anf of an orthonormal basis, which are very useful
for a rigorous mathematical foundation of QM.

A new  inner product, which is adequate for the physical interpretation, has been consistently introduced.

The case of Hamiltonians possessing complex eigenvalues is of an entirely different nature.
It arises in connection with models of dissipative or absorptive processes.

\end{document}